\documentclass[aps, twocolumn, superscriptaddress]{revtex4-1}

\usepackage{graphicx, amssymb, amsfonts, amsmath, amsthm, verbatim, url, hyperref, bbm, natbib}

\newcommand{\ket}[1]{| #1 \rangle}
\newcommand{\bra}[1]{\langle #1 |}

\def\u{\mathbf u}
\def\k{\mathbf k}
\def\V{\mathcal V}
\def\e{\mathrm{e}}
\def\i{\mathrm{i}}
\newcommand{\unity}{\ensuremath{\mathbbm 1}}

\theoremstyle{definition}
\newtheorem{lemma}{Lemma}

\newcommand{\bs}{$\bullet$}

\begin{document}
\title{Does causal dynamics imply local interactions?}

\author{Zolt\'an Zimbor\'as}
\affiliation{Wigner Research Centre for Physics, H-1121, Budapest, Hungary}
\affiliation{MTA-BME Lend\"ulet Quantum Information Theory Research Group, Budapest, Hungary}
\affiliation{Mathematical Institute, Budapest University of Technology and Economics,  H-1111, Budapest, Hungary}
\author{Terry Farrelly}
\affiliation{Institut f{\"u}r Theoretische Physik, Leibniz Universit{\"a}t, Appelstra{\ss}e 2, 30167 Hannover, Germany}
\affiliation{ARC Centre for Engineered Quantum Systems, School of Mathematics and Physics, University of Queensland, Brisbane, QLD 4072, Australia}
\author{Szil\'ard Farkas}
\affiliation{Wigner Research Centre for Physics, H-1121, Budapest, Hungary}
\author{Lluis Masanes}
\affiliation{Computer Science Department, University College London, United Kingdom.}
\affiliation{London Centre for Nanotechnology, University College London, United Kingdom.}
\date{\today}

\begin{abstract}
We consider quantum systems with causal dynamics in discrete spacetimes, also known as quantum cellular automata (QCA). 
Due to time-discreteness this type of dynamics is not characterized by a Hamiltonian but by a one-time-step unitary. This can be written as the exponential of a Hamiltonian but in a highly non-unique way. We ask if any of the Hamiltonians generating a QCA unitary is local in some sense, and we obtain two very different answers.
On one hand, we present an example of QCA for which all generating Hamiltonians are fully non-local, in the sense that interactions do not decay with the distance. 
We expect this result to have relevant consequences for the classification of topological phases in Floquet systems, given that this relies on the effective Hamiltonian.
On the other hand, we show that all one-dimensional quasi-free fermionic QCAs have quasi-local generating Hamiltonians, with interactions decaying exponentially in the massive case and algebraically in the critical case.
We also prove that some integrable systems do not have local, quasi-local nor low-weight constants of motion; a result that challenges the standard definition of integrability.
\end{abstract}
\maketitle

\section{Introduction}

Quantum cellular automata (QCAs) originally arose in the context of quantum computation as the generalization of classical cellular automata \cite{Feynman82} and were proven to be universal quantum computers \cite{Watrous95}.
QCAs can also be understood as the many-body generalization or ``second quantization" of quantum walks \cite{19Farrelly}.
From a physics perspective, QCAs are quantum field theories in discrete spacetimes obeying strict causality \cite{SW04, SFW06, 19Farrelly, 19Arrighi, DP14, BDPT15, GUWZ10}.
This means that after one time-step information only propagates a finite distance.  
Hence, QCAs provide a rigorous regularization of (continuous) quantum field theories which simultaneous preserve causality and unitarity, something impossible in Hamiltonian lattice field theory \cite{19Farrelly}.  In Lagrangian lattice field theory, the path integral is equivalent to a QCA for some field theories \cite{FS20}.
On a more speculative level, some arguments suggest that spacetime might be discrete at the Planck scale, and that all of the more familiar continuous spacetime physics emerges as an effective description at larger scales.
This opens the possibility of considering QCAs as Planck-scale theories.

The mathematical formulation of discrete-time quantum dynamics is different from that of continuous time.
In the discrete case, dynamics is represented by a one-time-step unitary evolution operator $W$ and in the continuous case by a Hamiltonian $H$.
The eigenstates of a Hamiltonian $H$ can be ordered with increasing energies, but the eigenstates of a unitary $W$ cannot be ordered because the corresponding quasi-energies are defined modulo $2\pi$.
This also makes unclear what should be the Gibbs states associated to $W$.
An exception to this are the unitaries that are close to the identity $W\approx \unity -\i\epsilon H$, which arise when continuous-time dynamics is Trotterized for simulations \cite{NielsenChuang}.
But despite the above-mentioned differences, it is reasonable to expect that, at large time scales, discrete-time models converge to continuous-time models.
The results presented in this work suggest that this convergence is not straightforward.

In this work we address the following question.
If $W$ is the evolution operator of a QCA, we consider all Hamiltonians $H$ which generate it via $W=\e^{-\i H}$, and ask whether one of these Hamiltonians is in some sense local.
In general, the Hamiltonian $H$ cannot have finite-range interactions, because the exponential of a finite-range $H$ is only approximately causal, as constrained by the Lieb-Robinson bound \cite{LR72}.
But $H$ can be local in the weaker sense of having interactions that decay with the distance.
In this work we present two extreme examples of QCAs in one spatial dimension, with opposite decaying behaviour.

The first model that we analyse (Section \ref{sec:fractal}) is a so called ``fractal QCA" introduced in \cite{GUWZ10}. 
We prove that any of its generating Hamiltonians has interactions which do not decay with the distance, and that the weight (number of qubits acted on) of the interaction terms is unbounded.
The implications of this are intriguing, as the effective Hamiltonian plays a key role in understanding topological phases of matter in Floquet
systems. But here, in contrast to our expectations, we see that the effective Hamiltonian can be extremely non-local. This leads to exciting questions, e.g., how does such non-locality impact our understanding of dynamical phases?

The evolution operator $W$ of the fractal QCA is a Clifford unitary \cite{SW04}, and these share some features with quasi-free bosonic unitaries, like the fact that dynamics can be represented in a symplectic phase space of dimension linear in the number of modes (qubits), which allows for efficiently simulating the dynamics with a classical computer.
But despite sharing these features with integrable systems, we prove in Section \ref{sec:integrab} that all conserved quantities of the fractal QCA (i.e.~operators that commute with $W$) are non-local and have unbounded weight (like the Hamiltonians).
This is a very interesting fact because it challenges one of the standard characterizations of integrable systems in terms of local (or low-weight) conserved quantities \cite{GME11}.
And suggest that, in the discrete-time scenario, integrability should be characterized differently.

The second family of QCAs that we analyze (Section \ref{sec:fermions}) have general quasi-free fermion dynamics in one spatial dimension.
In this case we show that there always exists a Hamiltonian with decaying interactions.
This decay is exponential in the gapped case and inversely proportional to the distance in the critical case.
However, we need to define what do we mean by gapped and critical when (quasi-)energy is defined modulo $2\pi$.
We prove that the whole algebra of operators corresponding to a type of quasi-particle drifts to the right at a constant speed equal to the winding number of the quasi-energy band associated to this quasi-particle.
Hence, when this winding number is not zero, the quasi-particle behaves like massless particles in quantum field theory.
For this reason we say that a quasi-free fermionic QCA is critical when some quasi-energy bands have non-zero winding number.
In contrast, when all winding numbers are zero, we say that the QCA is gapped.

\section{The fractal QCA}
\label{sec:fractal}

\subsection{Description of the model}


Clifford QCAs \cite{schlingemann2008structure, GUWZ10, haah2018nontrivial, haah2019clifford} are QCAs on lattices of qubits with the property that products of Pauli operators are mapped to products of Pauli operators.  In Ref. \cite{GUWZ10} one-dimensional Clifford QCAs were dived in two classes depending on the spacetime graph of the evolution of a single-site Pauli operator: those with a periodic structure ({\it periodic Clifford QCAs}) and those with a  spacetime graph that is self-similar over long timescales  ({\it fractal Clifford QCAs}).
This classification has since then been turned out to have importance in schemes of measurement-based quantum computation built on Clifford QCAs \cite{stephen2019subsystem} and the fractal property of various QCAs (Clifford and non-Clifford) has been studied intensely recently \cite{GNW10, yoshida2013exotic, hillberry2020entangled}.
In what follows we define a particular fractal Clifford QCA that was studied in \cite{GUWZ10}.

Consider a spin chain with $L$ qubits labelled by $r\in\{0,1, \ldots, L-1\}$ and periodic boundary conditions. 
We denote by $\sigma_x^r, \sigma_y^r, \sigma_z^r$ the Pauli sigma matrices acting on qubit $r$.
The evolution operator $W$ is determined by conditions
\begin{equation}\label{eq:exQCA}
 \begin{split}
  W^{\dagger}\sigma_z^rW &= \sigma_x^r\ ,\\
   W^{\dagger}\sigma_x^rW &= \sigma_x^{r-1}\sigma_y^r\sigma_x^{r+1}\ ,
 \end{split}
\end{equation}
for all $r$.
To see this, recall that $\sigma_y^r =\i\, \sigma_x^r \sigma_z^r$ and use $W^{\dagger}\sigma_y^r W = \i\, W^{\dagger}\sigma_x^r W W^{\dagger}\sigma_z^r W$.

To keep track of the evolution of a general $n$-qubit Pauli operator it is more convenient to use the phase-space description. 
Then each $n$-qubit Pauli operator $\sigma_\u$ is represented by a phase-space vector $\u = (q_0, p_0,  \ldots, q_{L-1}, p_{L-1} ) \in \{0,1\}^{2L}$ so that
\begin{align}
  \label{eq:Pauli element}
  \sigma_\u
  =
  \bigotimes_{r=0}^{L-1} 
  (\sigma_x^{r})^{q_r}
  (\sigma_z^{r})^{p_r} 
  \ .
\end{align}
With this notation we have that $\sigma_\u \sigma_{\u'} \propto \sigma_{\u+\u'}$ where addition in phase space is defined modulo 2.
We can write the phase space of the system as 
\begin{equation}
  \V = \bigoplus_{r=0}^{L-1} \V_r\ ,  
\end{equation}
where subspace $\V_r \cong \{0,1\}^{2}$ is associated to qubit $r$.
For each region of the lattice $\mathcal R \subseteq \{0, \ldots, L-1\}$ we define the corresponding subspace of Pauli operators
\begin{equation}
  \V_{\mathcal R} = 
  \bigoplus_{r\in \mathcal R} \V_r\ .
\end{equation}

Since Paulis form an operator basis we can express any of the Hamiltonians of the fractal QCA $W=\e^{-\i H}$ as
\begin{equation}\label{eq:H1}
  H = \sum_{\u\in \V} h_\u\, \sigma_\u\ ,
\end{equation}
where $h_\u$ are the real coefficients.
We say that a Hamiltonian is local if the coefficients $h_\u$ decay with the size of the support of $\u$.
That is, there is a monotonically decreasing function $f:\mathbb N \to \mathbb R^+$ such that
\begin{equation}\label{eq:def_loc}
  |h_\u| \leq f(D(\u))\ ,
\end{equation}
where the diameter of $\u$ is
\begin{equation}
  D(\u) = \min\{d(r_2,r_1):\u \in \V_{[r_1, r_2]}\}\ ,
\end{equation}
with
\begin{equation}
  d(r_2,r_1) = \begin{cases}
                r_2-r_1\ \mathrm{if}\ r_2-r_1\geq 0 \\
                r_2 -r_1 + L\ \mathrm{otherwise.}
               \end{cases}
\end{equation}

\subsection{No local Hamiltonians}
\label{sec:integrab}

The following lemma tells us that, for the fractal QCA defined in equation (\ref{eq:exQCA}), \emph{none} of the Hamiltonians satisfying $W=\e^{-\i H}$ is local in the sense of \eqref{eq:def_loc}.

\begin{lemma}\label{lemma 1}
Let $W$ be the QCA defined in equation (\ref{eq:exQCA}), and let $[r_1, r_2] \subseteq \{0, \ldots, L-1\}$ be any interval of the spin chain. For each Pauli operator in the interval $\u \in \V_{[r_1, r_2]}$ there is another Pauli operator in the larger interval $\u' \in \V_{[r_1-1, r_2+1]}$ which is not in any smaller interval and has the same coefficient $h_{\u'} = h_\u$ in \emph{any} of the Hamiltonians $H$ satisfying $W=\e^{-\i H}$, .  
\end{lemma}

\begin{proof}
Take \emph{any} Hamiltonian $H$, such that $W = e^{-iH}$.  Then, $H$ satisfies $W^{\dagger n}HW^n = H$ for any $n\in\mathbb{Z}$.  We can expand the Hamiltonian in terms of Pauli operators, rewriting equation (\ref{eq:H1}), as
\begin{equation}\label{eq:H}
 H=\sum_{\mathbf{u}\in\mathcal{O}}h_{\mathbf{u}}\sum_{n=0}^{m(\mathbf{u})-1}W^{\dagger n}\sigma_{\mathbf{u}}W^{n},
\end{equation}
where now $\mathcal{O}\subset \mathcal{V}$ is a set of labels of Pauli strings with one vector for \emph{each closed orbit} under applying $W$ up to $m(\mathbf{u})$ times.  As a convention, we can choose the $\mathbf{u}\in\mathcal{O}$ that labels each orbit to be the label corresponding to the Pauli string with the smallest support in the orbit (i.e., the Pauli string with the smallest $D(\mathbf{u})$).  Here $m(\mathbf{u})$ is the length of the orbit, meaning that $m(\mathbf{u})$ is the smallest positive integer such that
\begin{equation}
W^{\dagger m(\mathbf{u})}\sigma_{\mathbf{u}}W^{m(\mathbf{u})}=\sigma_{\mathbf{u}}.
\end{equation}
Note that $m(\mathbf{u})$ must exist for each orbit because we have a finite quantum system, and there are only a finite number of Pauli strings, so the orbits must be closed.

The next step is to show that terms like $W^{\dagger n}\sigma_{\mathbf{v}}W^{n}$ will generally spread over larger and larger regions for fractal QCAs.  Note that the region $\sigma_{\mathbf{v}}$ is supported on is not guaranteed to grow for \emph{only a single application} of $W$.  Take, for example, $\sigma_{\mathbf{v}}= \sigma_z^r\otimes \sigma_y^{r+1} \otimes \sigma_y^{r+2}\otimes \sigma^{r+3}_z$.  Then we have $W^{\dagger}\sigma_{\mathbf{v}}W = \sigma_y^{r+1}\otimes \sigma_y^{r+2}$.

We first note the fact that this QCA $W$ has no gliders, which is proved in appendix \ref{app:gliders}.  A glider is an operator $\sigma_{\mathbf{v}}$ with the property that $W^{\dagger}\sigma_{\mathbf{v}}W=S^{\dagger k}\sigma_{\mathbf{v}}S^k$, for some $k\in\mathbb{Z}$, where $S$ is just the unitary that shifts qubits one step to the left.  Next, we use lemma II.15 in \cite{GUWZ10}.  This shows that, if a Clifford QCA $W$ has no gliders then $W^n$ has no gliders for any $n\in\mathbb{Z}$.

Consider $\sigma_{\mathbf{v}}$ with support on an interval with $l$ sites with $l<L-2$.  If we consider how $\sigma_{\mathbf{v}}$ evolves over time, it follows that $W^{\dagger n}\sigma_{\mathbf{v}}W^n$ must eventually spread over an interval larger than $l$ sites.  To see this, we can argue by contradiction:\ suppose that that $W^{\dagger n}\sigma_{\mathbf{v}}W^n$ always remains localised on at most $l$ sites (the region may shift left or right but the range of sites is at most $l$).  Now there are only a finite number of Pauli strings on $l$ sites ($4^{l}$ many).  So at some value of $n\neq 0$ we must have
\begin{equation}
 W^{\dagger n}\sigma_{\mathbf{v}}W^n = S^{\dagger k}\sigma_{\mathbf{v}}S^k,
\end{equation}
for some $k\in\mathbb{Z}$.  But this means that $W^n$ has a glider, which is impossible because $W$ has no gliders.


Then we can apply this logic to the orbits in the sum in equation (\ref{eq:H}).  Therefore, we see that $H$ contains interactions between arbitrarily far regions \emph{with no decay of interaction strength} with distance.  More precisely, for any Pauli operator $\sigma_{\u}$ in the interval $\u \in \V_{[r_1, r_2]}$, there is another operator in the larger interval $\u' \in \V_{[r_1-1, r_2+1]}$ (and not in any smaller interval), which has the same coefficient $h_{\u'} = h_\u$ in any Hamiltonians $H$ satisfying $W=e^{-iH}$.
\end{proof}

\subsection{No local constants of motion
}

The proof of the above lemma not only applies to operators $H$ such that $W=\e^{-\i H}$ but to any operator which commutes with $W$. 
This shows that any constant of motion of the fractal QCA is non-local in the same sense that the Hamiltonians.


We finish this section by commenting about the continuous-time dynamics of any Hamiltonian $H$ that generates the fractal QCA $W=\e^{-\i H}$. 
As the lemma tells us, the continuous-time dynamics $\e^{-\i Ht}$ for $t\in \mathbb R$ is fully non-local but it has a particular type of destructive interference that cancels out all non-causal effects every time that $t$ reaches an integer value $t\in \mathbb Z$.

\section{Quasi-free fermionic QCAs}
\label{sec:fermions}


\subsection{Fermionic systems}

In this section we introduce some formalism for working with general quasi-free fermionic systems.
Consider $N$ fermionic modes with associated Majorana operators $a_i$ with $i\in \{1, 2, \ldots, 2N\}$.
These operators are Hermitian $a_i^\dagger =a_i$ and satisfy the canonical anti-commutation relations
\begin{equation}\label{cl}
  \left\{ a_i, a_{j}\right\}
  = 2\, \unity \delta_{ij}\ ,
\end{equation}
where $\unity$ is the identity and $\delta_{ij}$ the Kronecker-delta function. 
Instead of Majorana operators we could use the creation and annihilation ones 
\begin{align}\nonumber
  f_r^\dagger &= \frac 1 {\sqrt 2} \left( a_{2r} -\i\, a_{2r+1} \right) ,
  \\
  f_r &= \frac 1 {\sqrt 2} \left( a_{2r} +\i\, a_{2r+1} \right) ,
\end{align}
for $r=1,\ldots, N$, but Majoranas simplify our expressions with no loss of generality.

A unitary operator $W$ is quasi-free if it maps each Majorana operator onto a linear combinations of them
\begin{equation}\label{eq:qf}
 W^\dagger a_i W =\sum_{j} O_{ij}\, a_{j}\ .
\end{equation}
The Hermiticity of $a_i$ together with the anti-commutation relations \eqref{cl} imply that the matrix $O$ is orthogonal.
Imposing that $W$ commutes with the fermionic parity operator
\begin{equation}\label{eq:parity}
  Q = \bigotimes_r 
  (2 f_r^{\dagger} f_r -1 )\ ,
\end{equation}
implies (Lemma \ref{lem:parity} in appendix \ref{app:b}) that $O$ has unit determinant $O\in {\rm SO}(2N)$.
Note that particle-number conserving models (where $W$ commutes with $\sum_r f_r^{\dagger} f_r$) are a small subset of the quasi-free models.



A system of $N$ fermion modes can be mapped to $N$ qubits via the Jordan-Wigner transformation \cite{JW28}
\begin{eqnarray}
  \nonumber
  a_{2r} = (\mbox{$\prod_{s=1}^{r-1}$} \sigma_z^{s})\, \sigma_x^{r}\ ,
\\
  a_{2r+1} = (\mbox{$\prod_{s=1}^{r-1}$} \sigma_z^{s})\, \sigma_y^{r}\ ,
\end{eqnarray}
where $\sigma_{x,y,z}^{r}$ denote the Pauli sigma matrices acting on qubit $r$.
Although this representation is not local, the product of an even number of operators from $\{f_r, f^\dagger_{r}, f_{r+1}, f^\dagger_{r+1} \}$ only acts on qubits $r$ and $r+1$ for systems on an infinite line. 

For each $O\in {\rm SO}(2N)$ there is a (non-unique) real antisymmetric matrix $Z$ such that
$O=\e^{Z}$.
For any anti-symmetric matrix $A$ we define $\alpha(A)= \frac 1 4 \sum_{ij} A_{ij}a_i a_j$. From the anti-commutation relations \eqref{cl} it follows that $[\alpha(A),a_i]=-\sum_j A_{ij}a_j$.  Then using the identity
\begin{align}
\nonumber
 \e^{A}B\e^{-A} & =\e^{[A,\,\cdot\,]}B\\
 & = B+[A,B]+\frac{1}{2}[A,[A,B]]+\cdots
\end{align}
we arrive at
\begin{equation}
 \e^{-\alpha(Z)}a_i\,\e^{\alpha(Z)}=  \sum_{j}(\e^Z)_{ij}\, a_j\ .
\end{equation}
Therefore, up to a phase, any quasi-free unitary $W$ can be written as
\begin{equation}
 W=\e^{\alpha(Z)}\ .
\end{equation}
The Hamiltonian $H=\i\alpha(Z)$ is a possible generator for $W$.

\subsection{QCAs}

Consider a spin chain with $L$ sites labelled by $r \in \mathbb Z_L = \{0,1, \ldots, L-1 \}$ and periodic boundary conditions.
Each site $r$ contains $n$ fermionic modes represented by $2n$ Majorana operators $a_r^l$ with $l\in  \{1, \ldots, 2n\}$.
Complex linear combination of Majorana operators can be represented by vectors in $\mathbb C^L \otimes \mathbb C^{2n}$, where we separate the spatial $r$ and internal $l$ degrees of freedom.
The orthogonal matrix $O$ associated to the QCA's evolution operator $W$ via \eqref{eq:qf}  acts on the space $\mathbb C^L \otimes \mathbb C^{2n}$.

Let $|r\rangle$ be the orthonormal basis for $\mathbb C^L$ corresponding to the position.
Define the translation (or shift) operator $S$ acting on $\mathbb C^L$ via $S\, |r\rangle = |r+1\bmod L\rangle$.
The properties of translation invariance and causality imply that $O$ can be understood as the dynamics of a discrete-time quantum walk with coin space $\mathbb C^{2n}$.
It is well known (see, e.g., \cite{FS14}) that translation-invariance implies the structure
\begin{equation}\label{eq:decO}
  O = \sum_{q} S^q \otimes A_q\ ,
\end{equation}
where $q\in \mathbb Z_L$ and the operators $A_q$ act on the coin space  $\mathbb C^{2n}$.
In particular, the operator $A_q$ specifies how the information that is translated $q$ sites (to the right) is processed.
Additionally, causality enforces the existence of a neighborhood radius (also known as interaction range) $R$ beyond which information does not flow after only one timestep:
\begin{equation}\label{eq:caus}
  A_q = 0 \mbox{ for all } q \notin [-R,R]\ .
\end{equation}

\subsection{The Hamiltonian}

Next we obtain the spectral decomposition of \eqref{eq:decO}.
The eigenvectors of the translation operator
\begin{equation}\label{eq:Sk}
  S^q| k \rangle=\e^{-iqk}|k\rangle\ ,
\end{equation}
are the quasi-momentum states
\begin{equation}
 |k\rangle= \frac{1}{\sqrt{L}}\sum_{r=1}^{L}\e^{ikr}|r\rangle,
\end{equation}
with $k\in \frac {2\pi} L\, \{0,\ldots, L-1\}$. 
If we consider the ansatz $\ket{k} \otimes \ket{v}$ as an eigenvector of \eqref{eq:decO} then we obtain
\begin{equation}
  O\, |k\rangle \!\otimes\! |v\rangle
  = 
  |k\rangle \!\otimes\! 
  M_k |v\rangle\ ,   
\end{equation}
where we define
\begin{equation}\label{def M}
  M_k =\sum_q A_q\, \e^{-iqk}\ .
\end{equation}
The ansatz $\ket{k} \otimes \ket{v}$ is an eigenvector of $O$ if $\ket v$ is an eigenvector of $M_k$.

The orthogonality of \eqref{eq:decO} implies the unitarity of $M_k$ for all $k$.
Hence, the spectral decomposition
\begin{eqnarray}
  M_k 
  &=& 
  \sum_{s}
  \theta _k ^s\,
  P_k^s
  \ ,
\end{eqnarray}
has complex-phase eigenvectors $\theta_k^s$ and orthogonal spectral projectors $P^s_k$ labelled by $s$.



In all what follows on fermionic QCAs we work in the thermodynamic limit $L\to\infty$ to maximize the clarity of the results.
For finite $L$ the results are essentially the same but more cumbersome to express. 
We recall that in the thermodynamic limit locations are labelled by $r\in \mathbb Z$ and the momentum becomes continuous $k\in [0, 2\pi)$.

In this limit, the spectral decomposition of \eqref{eq:decO} can be written as
\begin{equation}
  \label{eq:spect O}
  O =
  \int_0^{2\pi}\! \frac {dk}{2\pi}\ 
  \sum_{s}
  \theta_k^s\
  |k \rangle\!\langle k| \otimes P_k^s
  \ .
\end{equation}
Causality \eqref{eq:caus} implies that the matrix \eqref{def M} is holomorphic in $k$ (as a complex variable $k\in \mathbb C$). 
Therefore we can choose the eigenvalues $\theta^s_k$ and spectral projections $P^s_k$ to be holomorphic in $k$ too (see theorem $1.10$ in chapter $2$ of \cite{Kato} and also \cite{Ahlbrecht12,Agaltsov18}).
Of course, we are primarily interested in the range $k\in[0,2\pi)$ as opposed to $k\in \mathbb C$.  

The matrix $M_k$ is periodic in  $k\in[0,2\pi)$, but this may not be true for an individual eigenvalue band $\theta_k^s$.  However, we can group bands together to form continuous and periodic energy bands in the following way.  
Suppose the  $n$ bands $s_1, s_2, \ldots, s_n$ form a closed curve
\begin{align}
  \label{cond:1}
  \lim_{k\to 2\pi}\theta_k ^{s_1} 
  &= \theta_0 ^{s_2}\ ,
  \\
  \lim_{k\to 2\pi}\theta_k ^{s_2} 
  &= \theta_0 ^{s_3}\ ,
  \\ \nonumber
  & \vdots
  \\ \label{cond:n}
  \lim_{k\to 2\pi}\theta_k ^{s_n} 
  &= \theta_0 ^{s_1}\ .
\end{align}
Then we can define the periodic holomorphic function $\Theta: [0,2\pi n) \to \mathbb C$ as
\begin{align}\label{def:Theta}
  \Theta(\k) = 
  \begin{cases}
    \theta_{\k}^{s_1} &\textrm{for}\ \k\in[0,2\pi)
    \\
    \theta_{\bf k -2\pi}^{s_2} &\textrm{for}\ \k\in[2\pi,4\pi)
    \\
    \ \vdots &\vdots
    \\
    \theta_{\k-2\pi(n -1)}^{s_n} &\textrm{for}\ \k\in[2\pi (n-1),2\pi n)
  \end{cases},
\end{align}
and the periodic holomorphic projector
\begin{align}
  \Pi (\k) = 
  \begin{cases}
    P_{\k}^{s_1} &\textrm{for}\ \k\in[0,2\pi)
    \\
    P_{\k-2\pi}^{s_2} &\textrm{for}\ \k \in[2\pi,4\pi)
    \\
    \ \vdots &\vdots
    \\
    P_{\k-2\pi(n -1)}^{s_n} &\textrm{for}\ \k\in[2\pi (n-1),2\pi n)
  \end{cases}.
\end{align}
Each closed curve is associated to a type of quasi-particle labelled by $\nu$. 
Quasi-particle $\nu$ is characterized by the objects $n^\nu, \Theta^\nu (\k), \Pi^\nu (\k)$, which also carry the label $\nu$.
This allows to write \eqref{eq:spect O} as
\begin{equation}
  O =
  \sum_\nu \int_0^{2\pi n^\nu}\! \frac {d\k}{2\pi}\ 
  \Theta^\nu(\k)\
  |\k \rangle\!\langle \k| 
  \otimes \Pi^\nu (\k)
  \ .
\end{equation}
The periodicity of the momentum eigenstates $\ket {k+2\pi} = \ket k$ allows to label them with the extended momentum $\k \in [0, 2\pi n^\nu)$.
The physical interpretation of the extended momentum is the following.
Conditions \eqref{cond:1}-\eqref{cond:n} for quasi-particle $\nu$ suggest that each lattice site $r\in \mathbb Z$ contains $n^\nu$ internal sites and that the dynamics $W$ enjoys a finer translational symmetry (for quasi-particle $\nu$) involving translations of fractional length $1/n^\nu$ which take into account the extra internal sites.
If the number $n^\nu$ is the same for all $\nu$ then we could fine-grain the lattice so that the new system has $n^\nu =1$ for all $\nu$, and the local number of modes is equal to the different types of quasi-particle (and no more).

Let $w^\nu$ denote the winding number of the periodic function $\Theta^\nu: [0,2\pi n^\nu) \to \mathbb C$.
This integer is the number of net loops around the unit circle in $\mathbb C$ that the function does across the interval $[0,2\pi n^\nu)$.
Because $\Theta^\nu (\k)$ is a holomorphic  function, there is another holomorphic function $E^\nu(\k)$ such that \begin{equation}
  \Theta^\nu (\k) 
  = 
  \exp\!\left(-\i E^\nu(\k)+\i\frac {w^\nu}{n^\nu} \k \right)\ . 
\end{equation}
Furthermore, in the range $\k\in [0,2\pi n^\nu)$ the function $E^\nu(\k)$ is periodic and takes real values. 
Note that these real values are not restricted to $[0, 2\pi)$ due to the continuity imposed in definition \eqref{def:Theta}.

Let us construct a Hamiltonian whose single-particle energy bands are the  functions $E^\nu(\k)$.
Our choice of matrix $Z$ satisfying $O= \e^Z$ is 
\begin{equation*}
  \i Z =
  \sum_\nu \int_0^{2\pi n^\nu}\!\! 
  \frac {d\k}{2\pi} 
  \left( E^\nu(\k) -\frac {w^\nu}{n^\nu} \k\right)
  |\k \rangle\!\langle \k| 
  \otimes \Pi^\nu (\k)
  \ .
\end{equation*}
The reason for writing the quasi-energies as the sum of two terms $(E^\nu(\k)-\frac {w^\nu}{n^\nu} \k)$ will be clear below.

Our choice of Hamiltonian $H$ satisfying $W= \e^{-\i H}$ is $H = \i \alpha(Z)$, which can be written as 
\begin{align}
  \nonumber
  H =
  &\sum_{\nu,r,r',l,l'} \int_0^{2\pi n^\nu}\!\! 
  \frac {d\k}{2\pi} 
  \left( E^\nu(\k) -\frac {w^\nu}{n^\nu} \k\right) \times
  \\ \label{def:H} & \times\,
  \e^{\i (r-r')\k}\,
  \bra l \Pi^\nu (\k) \ket {l'}\,
  a_r^l\, a_{r'}^{l'}
  \ .
\end{align}
In the next section we analyze the  locality of this Hamiltonian.

\subsection{Zero winding implies locality}

In this section we consider the case where all energy bands have zero winding number $w^\nu =0$.

In this case the coupling between lattice sites $r, r' \in \mathbb Z$ specified by Hamiltonian \eqref{def:H} is
\begin{align}
  \nonumber
  & \langle r,l |Z| r',l' \rangle
  \\ \label{eq:decay} =& -\i
  \sum_{\nu} \int_0^{2\pi n^\nu}\!\! 
  \frac {d\k}{2\pi}\, E^\nu(\k)\,
  \e^{\i (r-r')\k}\, 
  \bra l \Pi^\nu (\k) \ket {l'}\,
  \ ,
\end{align}
for any pair $l,l'$.
Recall that the functions $E^\nu(\k)$ and $\bra l \Pi^\nu (\k) \ket {l'}$ are analytic and periodic in the integration range $\k \in [0, 2\pi n^\nu)$.
Therefore, expression \eqref{eq:decay} is the Fourier transform of a periodic analytic function.
This is the premise of 
Lemma~\ref{lfou} from Appendix \ref{app:a}, which tells us that
\begin{equation}\label{e23}
  |\langle r,l |Z| r',l' \rangle| \leq
  C_1 \, \e^{-\beta_1 |r-r'|}\ , 
\end{equation}
for some constants $C_1, \beta_1 >0$.
That is, in the non-critical (gapped $H$) case interactions decay exponentially with the distance.

\subsection{Non-zero winding implies weak locality}

Suppose the band $\nu$ has non-zero winding number $w^\nu \neq 0$.
In the next subsection we see that this can be interpreted as the critical case, because any quasi-particle of type $\nu$ moves at constant speed irrespectively of its initial state.
Mimicking the behavior of massless particles in quantum field theory.

The contribution of quasi-particle $\nu$ to the interaction between sites $r,r' \in \mathbb Z$ in Hamiltonian \eqref{def:H} is
\begin{align}
  \int_0^{2\pi n^\nu}\!\! 
  \frac {d\k}{2\pi}
  \left( E^\nu(\k) -\frac {w^\nu}{n^\nu} \k\right)
  \e^{\i (r-r')\k}\, 
  \bra l \Pi^\nu (\k) \ket {l'}\,
  \ .
\end{align}
The part proportional to $E^\nu(\k)$ gives an exponential decay as in \eqref{e23}.
The part proportional to $w^\nu$ is the Fourier transform of the product of the analytic function $\bra l \Pi^\nu (\k) \ket {l'}$ times the discontinuous (on the  $[0, 2\pi n^\nu)$ torus) function $\k$. 
The Fourier transform of a product of two functions is the convolution of their Fourier transforms.
The Fourier transform of the analytic part  can be upper-bounded as
\begin{align}
  \left| \int_0^{2\pi n^\nu}\!\! 
  \frac {d\k}{2\pi}\,
  \e^{\i (r-r')\k}\, 
  \bra l \Pi^\nu (\k) \ket {l'}
  \right|
  \leq 
  C_2 \, \e^{-\beta_2 |r-r'|}\ .
\end{align}
And the Fourier transform of the discontinuous part is
\begin{equation}
  \int_0^{2\pi n^\nu}\!\! 
  \frac {d\k}{2\pi}\, 
  \frac {\k}{n^\nu}\,
  \e^{\i (r-r')\k} 
  = 
  \begin{cases}
  \frac{\i}{r'-r} 
  &\mbox{if } r\neq r' 
  \\
  \pi n^{\nu} 
  &\mbox{if } r=r'
  \end{cases}.
\end{equation}
Hence, their convolution can be upper-bounded by bounding the absolute value of each term
\begin{align}
  \nonumber
  &\left|
  \int_0^{2\pi n^\nu}\!\! 
  \frac {d\k}{2\pi}
  \frac {w^\nu}{n^\nu} \k
  \e^{\i (r-r')\k}\, 
  \bra l \Pi^\nu (\k) \ket {l'}
  \right|
  \\ \nonumber \leq\ &
  C_3\, \sum_{q\neq 0}\, 
  \frac{1}{|q|}\,
  \e^{-\beta_2 |r-r' -q|}
  +C_4\, \e^{-\beta_2 |r-r'|}
  \\ \leq\ & 
  \frac {C_5} {|r-r'|}\ ,
\end{align}
where $\beta_2, C_3, C_4, C_5$ are some constants.  
Hence, we conclude that when at least one of the functions $\Theta^\nu (\k)$ has non-zero winding number, the Hamiltonian $H$ involves interactions that decay no slower than the inverse of the distance.
This is a much weaker form of locality than the exponential decay \eqref{e23}.
As an example, the massless Dirac QCA \cite{D'Ariano12a,D'Ariano12b} has non-zero winding numbers and its Hamiltonian decays as $1/|r-r'|$ exactly.

For (single-particle) quantum walks with gapped spectra, effective quasi-local Hamiltonians were constructed in \cite{Osborne08}.  In contrast, we get bounds on the locality of quantum walk Hamiltonians with or without a gap.  Note, however, that \cite{Osborne08} had no assumption of translational invariance.

\subsection{Criticality as drift dynamics}

Let us consider a quasi-particle $\nu$ with non-zero winding number $w^\nu \neq 0$.
The corresponding sub-algebra of operators is generated by
\begin{equation}
  \sum_{r,l} \e^{\i \k r}\, 
  \bra l \Pi^\nu (\k) \ket v \,
  a_r^{l}\ ,
\end{equation}
for all $\ket v \in \mathbb C^{2n}$.
If the projector has rank one $\Pi^\nu (\k) = |v^\nu (\k)\rangle\! \langle v^\nu (\k)|$ then we can write the simpler expression
\begin{equation}
  b^\nu (\k)
  =
  \sum_{r,l} \e^{\i \k r}\, 
  \bra l v^\nu (\k) \rangle \,
  a_r^{l}\ ,
\end{equation}
for the generators of the sub-algebra of quasi-particle $\nu$.
By construction, the time evolution of these generators is
\begin{equation}
  W^\dagger  b^\nu (\k) W
  =
  \e^{-\i E^\nu(\k)}\, 
  \e^{\i\frac {w^\nu}{n^\nu} \k}\,
  b^\nu (\k)\ .
\end{equation}
The first phase $\e^{-\i E^\nu(\k)}$ corresponds the dynamics generated by an (exponentially) local Hamiltonian.
The second term $\e^{\i\frac {w^\nu}{n^\nu} \k}$ corresponds to a spatial translation $r \mapsto r- \frac{w^\nu} {n^\nu}$ for all the algebra of operators of quasi-particle $\nu$.
This implies that, irrespective of its initial state, the quasi-particle $\nu$ drifts at a constant speed $\frac {w^\nu} {n^\nu}$.
Mimicking the behavior of massless particles in quantum field theory.

\subsection{Remarks on free-fermion QCAs}

One approach to obtain quasi-free fermion QCAs is to take a quantum walk and apply fermionic second quantization \cite{19Farrelly}, resulting on a QCA that preserves particle number.
The family of quasi-free fermion QCAs that we consider is more general and includes QCAs which do not preserve particle number.


On another topic, the sum of all winding numbers is the index of the corresponding quantum walk
\begin{equation}
  \mathcal I = \sum_{\nu} w^\nu\ ,
\end{equation}
defined in~\cite{GNVW12}. 
This can be interpreted as the net amount of information flow along the chain. It is has been proven \cite{CGGSVWW16, asboth2012symmetries} that a quantum walk with non-zero index has gapless spectrum.

\section{QCAs generated by time-dependent Hamiltonians}

In \cite{ranard2020converse} it is proven that when a QCA $W$ has zero index (defined in \cite{GNVW12}) there always exists a time-dependent Hailtonian $H(t)$ with exponentially-decaying interactions which generates $W$ in a finite time $\tau$, that is
\begin{align}
  W= \mathcal T \e^{-\i \int_0^\tau H(t) dt}\ ,
\end{align}
where $\mathcal T$ is the time-ordering operator. 
This holds even if $W$ is an approximate QCA \cite{ranard2020converse}.
This implies that the fractal QCA (Section \ref{sec:fractal}), despite not having a quasi-local time-independent generator $H$, it has a quasi-local time-dependent generator $H(t)$.
However, in this work, we are concerned with time-independent Hamiltonians, because we want to relate QCAs to quantum field theories in high-energy and condensed matter physics.

It is worth mentioning that the definition of ``quasi-local Hamiltonian" in \cite{ranard2020converse} is different than ours. In \cite{ranard2020converse} a quasi-local Hamiltonian has exponentially-decaying interactions. Therefore, the fermionic QCAs with non-zero winding number (i.e. non-zero index) have non-quasi-local Hamiltonians.
In our work, a quasi-local Hamiltonian has interactions that decay with the distance in any way, no matter how slow. Therefore, the fact that the fractal QCA does not have a quasi-local Hamiltonian is a very strong result.

\section{Outlook}

This work gives rise to the following important open questions.
Do QCAs have a continuous-time limit?
How should we describe this limit? One possible approach is to work in the Hamiltonian picture, but our results suggest that this not always possible.
On another topic, how should we define integrable dynamics (as opposed to chaotic dynamics) in QCAs?
In the Hamiltonian picture integrability is defined in terms of the existence of local or low-weight constants of motion. In this work we have presented a QCA which should be considered integrable, because its dynamics can be described in phase space, but it does not enjoy local or low-weight constants of motion. This suggests that the integrability criterion for Hamiltonians is not applicable to QCAs.

\subsection*{Acknowledgment}

The authors would like to thank Tobias J.\ Osborne for useful discussions. 
LM acknowledges financial support by the UK's Engineering and Physical Sciences Research Council (grant number EP/R012393/1). 
TF was supported by the ERC grants QFTCMPS and SIQS, the cluster of excellence EXC201 Quantum Engineering and Space-Time Research, the DFG through SFB 1227 (DQ-mat), and the Australian Research Council  Centres of Excellence for Engineered Quantum Systems (EQUS, CE170100009). ZZ acknowledges support from the J\'anos Bolyai Research Scholarship, the UKNP Bolyai+ Grant, and the NKFIH Grants No. K124152, K124176 KH129601, K120569, and from the Hungarian Quantum Technology National Excellence Program, Project No. 2017-1.2.1-NKP-2017-00001.

\bibliographystyle{plainnat}
\bibliography{QCA}

\begin{thebibliography}{35}
\providecommand{\natexlab}[1]{#1}
\providecommand{\url}[1]{\texttt{#1}}
\expandafter\ifx\csname urlstyle\endcsname\relax
  \providecommand{\doi}[1]{doi: #1}\else
  \providecommand{\doi}{doi: \begingroup \urlstyle{rm}\Url}\fi

\bibitem[Agaltsov(2018)]{Agaltsov18}
A.~Agaltsov.
\newblock Eigenvalues of analytic families of operators, 2018.
\newblock URL \url{https://www2.mps.mpg.de/homes/
  agaltsov/notes/holomkato76.html}.

\bibitem[Ahlbrecht(2012)]{Ahlbrecht12}
A.~Ahlbrecht.
\newblock \emph{Asymptotic behavior of decoherent and interacting quantum
  walks}.
\newblock PhD thesis, Leibniz Universit{\"a}t Hannover, 2012.
\newblock URL \url{https://doi.org/10.1007/s11128-012-0389-4}.

\bibitem[Arrighi(2019)]{19Arrighi}
P.~Arrighi.
\newblock An overview of quantum cellular automata.
\newblock \emph{Natural Computing}, 2019.
\newblock \doi{https://doi.org/10.48550/arXiv.1904.12956}.
\newblock https://doi.org/10.1007/s11047-019-09762-6.

\bibitem[Bisio et~al.(2015)Bisio, D'Ariano, Perinotti, and Tosini]{BDPT15}
A.~Bisio, G.~M. D'Ariano, P.~Perinotti, and A.~Tosini.
\newblock {Weyl, Dirac and Maxwell Quantum Cellular Automata}: {Analytical
  Solutions and Phenomenological Predictions of the Quantum Cellular Automata
  Theory of Free Field}.
\newblock \emph{Found. Phys.}, 45\penalty0 (10):\penalty0 1203--1221, 2015.
\newblock \doi{https://doi.org/10.1007/s10701-015-9927-0}.

\bibitem[Brezis and Nirenberg(1995)]{BN95}
H.~Brezis and L.~Nirenberg.
\newblock Degree theory and {BMO}; part {I}: {C}ompact manifolds without
  boundaries.
\newblock \emph{Selecta Mathematica}, 1\penalty0 (2):\penalty0 197--263, 1995.
\newblock ISSN 1022-1824.
\newblock \doi{http://dx.doi.org/10.1007/BF01671566}.

\bibitem[Cedzich et~al.(2018)Cedzich, Geib, Gr{\"u}nbaum, Stahl, Vel{\'a}zquez,
  Werner, and Werner]{CGGSVWW16}
C.~Cedzich, T.~Geib, F.~A. Gr{\"u}nbaum, C.~Stahl, L.~Vel{\'a}zquez, A.~H.
  Werner, and R.~F. Werner.
\newblock The topological classification of one-dimensional symmetric quantum
  walks.
\newblock \emph{Annales Henri Poincar{\'e}}, 19\penalty0 (2):\penalty0
  325--383, 2018.
\newblock \doi{https://doi.org/10.1007/s00023-017-0630-x}.

\bibitem[D.~J.~Shepherd(2006)]{SFW06}
R.~F.~Werner D.~J.~Shepherd, T.~Franz.
\newblock Universally programmable quantum cellular automaton.
\newblock \emph{Phys. Rev. Lett.}, 97:\penalty0 020502, 2006.
\newblock \doi{https://doi.org/10.1103/PhysRevLett.97.020502}.

\bibitem[D'Ariano(2012{\natexlab{a}})]{D'Ariano12a}
G.~M. D'Ariano.
\newblock The quantum field as a quantum computer.
\newblock \emph{Physics Letters A}, 376\penalty0 (5):\penalty0 697--702,
  2012{\natexlab{a}}.
\newblock \doi{http://dx.doi.org/10.1016/j.physleta.2011.12.021}.

\bibitem[D'Ariano(2012{\natexlab{b}})]{D'Ariano12b}
G.~M. D'Ariano.
\newblock Physics as quantum information processing:\ quantum fields as quantum
  automata.
\newblock \emph{Foundations of Probability and Physics - 6, AIP Conf. Proc.},
  page 1424 371, 2012{\natexlab{b}}.
\newblock \doi{https://doi.org/10.48550/arXiv.1110.6725}.

\bibitem[D'Ariano and Perinotti(2014)]{DP14}
G.~M. D'Ariano and P.~Perinotti.
\newblock Derivation of the {D}irac equation from principles of information
  processing.
\newblock \emph{Phys. Rev. A}, 90:\penalty0 062106, 2014.
\newblock \doi{https://doi.org/10.1103/PhysRevA.90.062106}.

\bibitem[Farrelly and Streich(2020)]{FS20}
T.~Farrelly and J.~Streich.
\newblock {D}iscretizing quantum field theories for quantum simulation.
\newblock \emph{arxiv preprint arXiv:2002.02643}, 2020.
\newblock \doi{https://doi.org/10.48550/arXiv.2002.02643}.

\bibitem[Farrelly(2020)]{19Farrelly}
T.~C. Farrelly.
\newblock A review of {Q}uantum {C}ellular {A}utomata.
\newblock \emph{{Quantum}}, 4:\penalty0 368, 2020.
\newblock ISSN 2521-327X.
\newblock \doi{https://doi.org/10.22331/q-2020-11-30-368}.

\bibitem[Farrelly and Short(2014)]{FS14}
T.~C. Farrelly and A.~J. Short.
\newblock Discrete spacetime and relativistic quantum particles.
\newblock \emph{Phys. Rev. A}, 89:\penalty0 062109, 2014.
\newblock \doi{https://doi.org/10.1103/PhysRevA.89.062109}.

\bibitem[Feynman(1982)]{Feynman82}
R.~P. Feynman.
\newblock Simulating physics with computers.
\newblock \emph{International Journal of Theoretical Physics}, 21:\penalty0
  467--488, 1982.
\newblock ISSN 0020-7748.
\newblock \doi{https://doi.org/10.1007/BF02650179}.

\bibitem[Gogolin et~al.(2011)Gogolin, M{\"u}ller, and Eisert]{GME11}
C.~Gogolin, M.~P. M{\"u}ller, and J.~Eisert.
\newblock Absence of thermalization in nonintegrable systems.
\newblock \emph{Phys. Rev. Lett.}, 106:\penalty0 040401, 2011.
\newblock \doi{https://doi.org/10.1103/PhysRevLett.106.040401}.

\bibitem[Gross et~al.(2012)Gross, Nesme, Vogts, and Werner]{GNVW12}
D.~Gross, V.~Nesme, H.~Vogts, and R.~F. Werner.
\newblock Index theory of one dimensional quantum walks and cellular automata.
\newblock \emph{Communications in Mathematical Physics}, 310:\penalty0
  419--454, 2012.
\newblock ISSN 0010-3616.
\newblock \doi{http://dx.doi.org/10.1007/s00220-012-1423-1}.

\bibitem[G{\"u}tschow et~al.(2010{\natexlab{a}})G{\"u}tschow, Nesme, and
  Werner]{GNW10}
J.~G{\"u}tschow, V.~Nesme, and R.~F. Werner.
\newblock The fractal structure of cellular automata on abelian groups.
\newblock \emph{{Discrete Mathematics \& Theoretical Computer Science}}, {DMTCS
  Proceedings vol. AL, Automata 2010 - 16th Intl. Workshop on CA and DCS},
  2010{\natexlab{a}}.
\newblock \doi{https://doi.org/10.46298/dmtcs.2759}.

\bibitem[G{\"u}tschow et~al.(2010{\natexlab{b}})G{\"u}tschow, Uphoff, Werner,
  and Zimbor\'as]{GUWZ10}
J.~G{\"u}tschow, S.~Uphoff, R.~F. Werner, and Z.~Zimbor\'as.
\newblock Time asymptotics and entanglement generation of clifford quantum
  cellular automata.
\newblock \emph{Journal of Mathematical Physics}, 51\penalty0 (1):\penalty0
  015203, 2010{\natexlab{b}}.
\newblock \doi{http://dx.doi.org/10.1063/1.3278513}.

\bibitem[Haah(2021)]{haah2019clifford}
J.~Haah.
\newblock Clifford quantum cellular automata: Trivial group in 2d and witt
  group in 3d.
\newblock \emph{J. Math. Phys.}, 62:\penalty0 092202, 2021.
\newblock \doi{https://doi.org/10.1063/5.0022185}.

\bibitem[Haah et~al.(2018)Haah, Fidkowski, and Hastings]{haah2018nontrivial}
J.~Haah, L.~Fidkowski, and M.~B. Hastings.
\newblock Nontrivial quantum cellular automata in higher dimensions.
\newblock \emph{arXiv preprint arXiv:1812.01625}, 2018.
\newblock \doi{https://doi.org/10.48550/arXiv.1812.01625}.

\bibitem[Hillberry et~al.(2020)Hillberry, Jones, Vargas, Rall, Halpern, Bao,
  Notarnicola, Montangero, and Carr]{hillberry2020entangled}
L.~E. Hillberry, M.~T. Jones, D.~L. Vargas, P.~Rall, N.~Y. Halpern, N.~Bao,
  S.~Notarnicola, S.~Montangero, and L.~D. Carr.
\newblock Entangled quantum cellular automata, physical complexity, and
  {G}oldilocks rules.
\newblock \emph{Quantum Sci. Technol.}, 6:\penalty0 045017, 2020.
\newblock \doi{https://doi.org/10.1088/2058-9565/ac1c41}.

\bibitem[Jordan and Wigner(1928)]{JW28}
P.~Jordan and E.~Wigner.
\newblock \"{U}ber das {P}aulische \"{A}quivalenzverbot.
\newblock \emph{Z. Physik}, 47:\penalty0 631--651, 1928.
\newblock URL \url{http://link.springer.com/article/10.1007/BF01331938}.

\bibitem[Kato(1995)]{Kato}
T.~Kato.
\newblock \emph{Perturbation theory for linear operators}.
\newblock Springer, 2nd edition, 1995.
\newblock \doi{https://doi.org/10.1007/978-3-642-66282-9}.

\bibitem[Krantz and Parks(2002)]{KP02}
S.~G. Krantz and H.~R. Parks.
\newblock \emph{A Primer of Real Analytic Functions}.
\newblock A Primer of Real Analytic Functions. Birkh{\"a}user Boston, 2002.
\newblock ISBN 9780817642648.
\newblock \doi{https://doi.org/10.1007/978-0-8176-8134-0}.

\bibitem[Lieb and Robinson(1972)]{LR72}
E.~H. Lieb and D.~W. Robinson.
\newblock The finite group velocity of quantum spin systems.
\newblock \emph{Communications in Mathematical Physics}, 28\penalty0
  (3):\penalty0 251--257, 1972.
\newblock \doi{https://doi.org/10.1007/BF01645779}.

\bibitem[Nielsen and Chuang(2010)]{NielsenChuang}
M.~A. Nielsen and I.~L. Chuang.
\newblock \emph{Quantum {C}omputation and {Q}uantum {I}nformation}.
\newblock Cambridge University Press, 2010.
\newblock \doi{https://doi.org/10.1017/CBO9780511976667f}.

\bibitem[Offner()]{Offner}
C.~D. Offner.
\newblock A little harmonic analysis.
\newblock \url{http://www.cs.umb.edu/~offner/files/harm_an.pdf}.

\bibitem[Osborne(2008)]{Osborne08}
T.~J. Osborne.
\newblock Approximate locality for quantum systems on graphs.
\newblock \emph{Phys. Rev. Lett.}, 101:\penalty0 140503, 2008.
\newblock \doi{https://doi.org/10.1103/PhysRevLett.101.140503}.

\bibitem[Ranard et~al.(2020)Ranard, Walter, and Witteveen]{ranard2020converse}
D.~Ranard, M.~Walter, and F.~Witteveen.
\newblock A converse to {L}ieb-{R}obinson bounds in one dimension using index
  theory.
\newblock \emph{arxiv preprint arXiv:2012.00741}, 2020.
\newblock \doi{https://doi.org/10.48550/arXiv.2012.00741}.

\bibitem[Schlingemann et~al.(2008)Schlingemann, Vogts, and
  Werner]{schlingemann2008structure}
D.~M. Schlingemann, H.~Vogts, and R.~F. Werner.
\newblock On the structure of clifford quantum cellular automata.
\newblock \emph{Journal of Mathematical Physics}, 49\penalty0 (11):\penalty0
  112104, 2008.
\newblock \doi{https://doi.org/10.1063/1.3005565}.

\bibitem[Schumacher and Werner(2004)]{SW04}
B.~Schumacher and R.~F. Werner.
\newblock Reversible {Q}uantum {C}ellular {A}utomata.
\newblock \emph{arxiv preprint arXiv:0405174}, 2004.
\newblock \doi{https://doi.org/10.48550/arXiv.quant-ph/0405174}.

\bibitem[Stephen et~al.(2019)Stephen, Nautrup, Bermejo-Vega, Eisert, and
  Raussendorf]{stephen2019subsystem}
D.~T. Stephen, H.~P. Nautrup, J.~Bermejo-Vega, J.~Eisert, and R.~Raussendorf.
\newblock Subsystem symmetries, quantum cellular automata, and computational
  phases of quantum matter.
\newblock \emph{Quantum}, 3:\penalty0 142, 2019.
\newblock \doi{https://doi.org/10.22331/q-2019-05-20-142}.

\bibitem[Tarasinski et~al.(2014)Tarasinski, Asb{\'o}th, and
  Dahlhaus]{asboth2012symmetries}
B.~Tarasinski, J.~K. Asb{\'o}th, and J.~P. Dahlhaus.
\newblock Scattering theory of topological phases in discrete-time quantum
  walks.
\newblock \emph{Physical Review A}, 89\penalty0 (4):\penalty0 042327, 2014.
\newblock \doi{https://doi.org/10.1103/PhysRevA.89.042327}.

\bibitem[Watrous(1995)]{Watrous95}
J.~Watrous.
\newblock On one-dimensional quantum cellular automata.
\newblock \emph{Proceedings of the 36th Annual IEEE Symposium on Foundations of
  Computer Science}, pages 528--537, 1995.
\newblock \doi{https://doi.org/10.1109/sfcs.1995.492583}.

\bibitem[Yoshida(2013)]{yoshida2013exotic}
B.~Yoshida.
\newblock Exotic topological order in fractal spin liquids.
\newblock \emph{Physical Review B}, 88\penalty0 (12):\penalty0 125122, 2013.
\newblock \doi{https://doi.org/10.1103/PhysRevB.88.125122}.

\end{thebibliography}

\appendix

\section{The fractal QCA has no gliders}
\label{app:gliders}
To simplify notations, we shall visualize the action of $W$ on monomials of Pauli-matrices as follows. The matrices $\sigma_x$, $\sigma_y$, and $\sigma_z$ themselves will be denoted by x, y, and z, respectively. Monomials of the Pauli matrices obtained by successive applications of $W$ will be written under each other. For example, the action of $W$ on the sigma matrices can be written schematically as

\medskip
\begin{tabular}{c@{\;}c@{\;}c@{\;}}
  &x& \\
x&y&x
\end{tabular}
\begin{tabular}{c@{\;}c@{\;}c@{\;}}
  &y& \\
x&z&x
\end{tabular}
\begin{tabular}{c@{\;}}
z\phantom{,}\\
x,
\end{tabular}

\medskip
\noindent and the inverse of $W$ is 

\medskip
\begin{tabular}{c@{\;}c@{\;}c@{\;}}
  &y& \\
z&x&z
\end{tabular}
\begin{tabular}{c@{\;}c@{\;}c@{\;}}
  &z& \\
z&y&z
\end{tabular}
\begin{tabular}{c@{\;}}
x\phantom{.}\\
z.
\end{tabular}

\medskip
\noindent Since they are inessential for the argument, signs are ignored by our notations, so a monomial is denoted by the same string as minus the same monomial. 

Our proof for the absence of gliders consists of three steps. First we show that if there were gliders, their length could not change under the time evolution. Then we demonstrate that this property makes it possible to define even simpler gliders, which we will call rigid gliders. Finally, a case by case study rules out the existence of rigid gliders.

\medskip
(1) The neighborhood of the left end of a string evolves in the following way:

\medskip
\begin{tabular}{c@{\;}c@{\;}c@{\;}c@{\;}c@{\;}c@{\;}}
z&\bs&\bs&\bs&\bs&\bs\\
  &z   &\bs&\bs&\bs&\bs\\
  &     &z   &\bs&\bs&\bs\\
  &     &     &y   &\bs&\bs\\
  &     &x   &\bs&\bs&\bs\\
  &x   &\bs&\bs&\bs&\bs\\
x&\bs&\bs&\bs&\bs&\bs\    
\end{tabular}
\begin{tabular}{c@{\;}c@{\;}c@{\;}c@{\;}c@{\;}c@{\;}c@{\;}}
\bs&\bs&\bs&\bs&\bs&\bs\\
  &\bs&\bs&\bs&\bs&\bs\\
  &     &\bs &\bs&\bs&\bs\\
  &     &     &x   &\bs&\bs\\
  &     &x   &\bs&\bs&\bs\\
  &x   &\bs&\bs&\bs&\bs\\
x&\bs&\bs&\bs&\bs&\bs\\
\end{tabular}
\begin{tabular}{c@{\;}c@{\;}c@{\;}c@{\;}c@{\;}c@{\;}c@{\;}c}
z&\bs&\bs&\bs&\bs&\bs&\\
  &z   &\bs&\bs&\bs&\bs&\\
  &     &z   &\bs&\bs&\bs&\\
  &     &     &z   &\bs&\bs&$\Longleftarrow$\\
  &     &\bs&\bs&\bs&\bs&\\
  &\bs&\bs&\bs&\bs&\bs&\\
\bs&\bs&\bs&\bs&\bs&\bs&\\
\end{tabular}
\medskip

\noindent The arrow indicates the time at which we specify the leftmost operator, and the three patterns correspond to the three possible choices. The bullets (\bs) indicates that any of the operators x, y, or z may be assigned to the particular position, but it is also possible that the identity is assigned to it. So where the frontier consists of bullets, we actually do not know exactly where the frontier lies. Similarly, the neighborhood of the right frontier looks like this:

\medskip
\begin{tabular}{c@{\;}c@{\;}c@{\;}c@{\;}c@{\;}c@{\;}}
\bs&\bs&\bs&\bs&\bs&z\\
\bs&\bs&\bs&\bs&z&\phantom{z}\\
\bs&\bs&\bs&z&\phantom{z}&\phantom{z}\\
\bs&\bs&y&\phantom{z}&\phantom{z}&\phantom{z}\\
\bs&\bs&\bs&x&\phantom{z}&\phantom{z}\\
\bs&\bs&\bs&\bs&x&\phantom{z}\\
\bs&\bs&\bs&\bs&\bs&x\\
\end{tabular}
\begin{tabular}{c@{\;}c@{\;}c@{\;}c@{\;}c@{\;}c@{\;}}
\bs&\bs&\bs&\bs&\bs&\bs\\
\bs&\bs&\bs&\bs&\bs&\phantom{z}\\
\bs&\bs&\bs&\bs&\phantom{z}&\phantom{z}\\
\bs&\bs&x&\phantom{z}&\phantom{z}&\phantom{z}\\
\bs&\bs&\bs&x&\phantom{z}&\phantom{z}\\
\bs&\bs&\bs&\bs&x&\phantom{z}\\
\bs&\bs&\bs&\bs&\bs&x\\
\end{tabular}
\begin{tabular}{c@{\;}c@{\;}c@{\;}c@{\;}c@{\;}c@{\;}c}
\bs&\bs&\bs&\bs&\bs&z&\\
\bs&\bs&\bs&\bs&z&\phantom{z}&\\
\bs&\bs&\bs&z&\phantom{z}&\phantom{z}&\\
\bs&\bs&z&\phantom{z}&\phantom{z}&\phantom{z}&$\Longleftarrow$\\
\bs&\bs&\bs&\bs&\phantom{z}&\phantom{z}&\\
\bs&\bs&\bs&\bs&\bs&\phantom{z}&\\
\bs&\bs&\bs&\bs&\bs&\bs&\\
\end{tabular}

\medskip
We would like to match the left and right frontiers so that they enclose the evolution of a glider. For example, take a glider whose left frontier contains a z. The past of this site is represented by a z-sequence extending indefinitely in the left upward direction. On the right hand side, this has to be matched with a parallel x-sequence, which in turn implies that the z-sequence on the left hand side has to extend {\it ad infinitum} in the future as well. This way we obtain a right-moving glider. Assuming that the left frontier contains an x, we get a left-moving glider. Clearly, no glider can have a y in the frontier. So the only possible gliders evolve as

\medskip
\hspace{2ex}x $\mathrm{S}_1$ z\hspace{4ex}z $\mathrm{S}_1$ x

\hspace{1ex}x $\mathrm{S}_2$ z\hspace{6ex}z $\mathrm{S}_2$ x

x $\mathrm{S}_3$ z\hspace{8ex}z $\mathrm{S}_3$ x,

\medskip
\noindent where $S_1, S_2, \dots$ are potentially different strings of equal length. 

\medskip
(2) Suppose that there is a glider for which $\mathrm{S}_1\neq\mathrm{S}_2$. Then there is also a shorter glider. Consider for example a left-moving glider. Evolve it by one time step, translate it by one lattice site to the right, and multiply the result by the original operator. If $*$ denotes the product of the corresponding operators, so that x$*$x, y$*$y, and z$*$z are the empty strings, then the new glider evolves in the following way:

\medskip
\hspace{2ex}(x$*$x) $(\mathrm{S}_1\!*\!\mathrm{S}_2)$ (z$*$z)\phantom{=}\hspace{4ex}$(\mathrm{S}_1\!*\!\mathrm{S}_2)$

\hspace{1ex}(x$*$x) $(\mathrm{S}_2\!*\!\mathrm{S}_3)$ (z$*$z)\hspace{2ex}=\hspace{2ex}$(\mathrm{S}_2\!*\!\mathrm{S}_3)$

(x$*$x) $(\mathrm{S}_3\!*\!\mathrm{S}_4)$ (z$*$z)\phantom{=}\hspace{4ex}$(\mathrm{S}_3\!*\!\mathrm{S}_4)$.

\medskip
\noindent So the new glider is shorter at least by two than the original one. Since the evolution of this glider must follow the previously derived pattern, the left and right frontiers are  again x and z, respectively. We can redefine $\mathrm{S}_i$ as the strings representing the inner part of the new glider. If these strings change in time, we can repeat the procedure, thereby obtaining an even shorter glider. After some iterations, we eventually get a rigid glider, for which $\mathrm{S}_1 = \mathrm{S}_2$, so it is simply translated by the time evolution.

\medskip
(3) To rule out rigid gliders, it is enough to look at the operator to the right of the frontier and see what we get after one time step:

\medskip
\begin{tabular}{c@{\;}c@{\;}c@{\;}c}
\phantom{x}&x&x&\bs\\
x&z&\bs&\bs
\end{tabular}
\begin{tabular}{c@{\;}c@{\;}c@{\;}c}
\phantom{x}&x&y&\bs\\
x&z&\bs&\bs
\end{tabular}
\begin{tabular}{c@{\;}c@{\;}c@{\;}c}
\phantom{x}&x&z&\bs\\
x&y&\bs&\bs
\end{tabular}
\begin{tabular}{c@{\;}c@{\;}c@{\;}c}
\phantom{x}&x&\phantom{x}&\bs\\
x&y&\bs&\bs
\end{tabular}

\medskip
\noindent None of these are rigid. An argument analogous to (2) and (3) shows that there are no right-moving gliders either. (This already follows from the non-existence of left-moving gliders because the time evolution is invariant under reflection.)

\section{Fourier lemmas}\label{app:a}

The following two lemmas are proved in References~\cite{KP02,Offner} and~\cite{BN95} respectively.

\begin{lemma}\label{lfou} 
Let $f: \mathcal{T \to T}$ be an analytic function on the torus with Fourier decomposition
\begin{equation}
  f(k) = 
  \sum_{r\in \mathbb Z} \hat f(r)\, \e^{\i r k}\ .
\end{equation}
Then there exist two positive constants $C, \beta >0$ such that 
\begin{equation}
  |\hat f(r)| < C\, e^{-\beta |r|}\ .
\end{equation}
\end{lemma}

\begin{lemma}\label{FTWN}
If $f: \mathcal{T \to T}$ is an analytic function with winding number $w$ then
\begin{equation}
  \sum_{r \in \mathbb Z} r\, |\hat f(r)|^2 =
  w\ .
\end{equation}
\end{lemma}

\section{Determinant of $O$}\label{app:b}

\begin{lemma}\label{lem:parity}
A quasi-free fermionic unitary commutes with the parity operator \eqref{eq:parity} if and only if the corresponding orthogonal matrix \eqref{eq:qf} has unity determinant.
\end{lemma}

\begin{proof}
First, note that the parity operator \eqref{eq:parity} can be written as
\begin{equation}
  Q = a_1 a_2 \cdots a_{2N}\ .
\end{equation}
Second, note that
\begin{equation*}
  W^{\dagger} a_1 \cdots a_{2N} W
  = \sum_{i_1,\ldots, i_{2N}}
  O_{1,i_1} \cdots O_{2N,i_{2N}}\,
  a_{i_1} \cdots a_{i_{2N}}.
\end{equation*}
Using the anti-commutation relations we can write
\begin{equation}
  a_{i_1} \cdots a_{i_{2N}}
  = \varepsilon_{i_1 \ldots i_{2N}}
  a_1 \cdots a_{2N}\ ,
\end{equation}
where $\varepsilon_{i_1 \ldots i_{2N}}$ is the Levi-Civita symbol.
Combining all of the above we obtain
\begin{equation}
  1= 
  \sum_{i_1,\ldots, i_{2N}}
  O_{1,i_1} \cdots O_{2N,i_{2N}}\,
  \varepsilon_{i_1 \ldots i_{2N}}
  = \textrm{det}(O)\ .
\end{equation}
\end{proof}

\end{document}